\documentclass[preprint,10pt]{elsarticle}

\makeatletter
\def\ps@pprintTitle{%
\let\@oddhead\@empty
\let\@evenhead\@empty
\def\@oddfoot{\centerline{\thepage}}%
\let\@evenfoot\@oddfoot}
\makeatother

\journal{Physica D}

\usepackage{graphicx,bbm,setspace,amssymb,amsfonts,amsmath,mathtools,mathrsfs,stmaryrd,amsthm,bm,etoolbox,comment,hyperref,url,caption,subcaption,color,enumitem,algorithm,booktabs,microtype,nicefrac,lingmacros,nccmath,tree-dvips,tikz-cd}
\usepackage[noend]{algpseudocode}
\usepackage{algorithm}

\usepackage[colorinlistoftodos]{todonotes}
\usepackage[utf8]{inputenc} 
\usepackage[T1]{fontenc}    

\DeclareMathOperator*{\loggrad}{log-grad}

\usepackage[noend]{algpseudocode}

\newtheorem{theorem}{Theorem}[section]

\newtheorem{thm-defn}[theorem]{Theorem/Definition}

\newtheorem{prop}[theorem]{Proposition}

\theoremstyle{definition}

\theoremstyle{remark}

\newcommand{\ignore}[1]{}{}

\begin{document}

\begin{frontmatter}

\title{Comparing the dynamics of COVID-19 infection and mortality in the United States, India, and Brazil}
   
\author[label1]{Nick James} \ead{nick.james@unimelb.edu.au}
\author[label2]{Max Menzies} \ead{max.menzies@alumni.harvard.edu}
\author[label1]{Howard Bondell}
\address[label1]{School of Mathematics and Statistics, University of Melbourne, Victoria, Australia}
\address[label2]{Beijing Institute of Mathematical Sciences and Applications, Tsinghua University, Beijing, China}

\begin{abstract}
This paper compares and contrasts the spread and impact of COVID-19 in the three countries most heavily impacted by the pandemic: the United States (US), India and Brazil. All three of these countries have a federal structure, in which the individual states have largely determined the response to the pandemic. Thus, we perform an extensive analysis of the individual states of these three countries to determine patterns of similarity within each. First, we analyse structural similarity and anomalies in the trajectories of cases and deaths as multivariate time series. Next, we study the lengths of the different waves of the virus outbreaks across the three countries and their states. Finally, we investigate suitable time offsets between cases and deaths as a function of the distinct outbreak waves. In all these analyses, we consistently reveal more characteristically distinct behaviour between US and Indian states, while Brazilian states exhibit less structure in their wave behaviour and changing progression between cases and deaths.

\end{abstract}


\end{frontmatter}

\section{Introduction}
\label{sec:Introduction}

The United States (US), India and Brazil have each been severely impacted by COVID-19 and lead the world in both case and death counts. While the three countries have quite different cultures and levels of economic and technological development, they each have a similar federation structure, with governing responsibilities divided between federal and state governments. In all three countries, government responses have consistently differed between constituent states and over time \cite{Haffajee2020,daSilva2021,india}, yielding different levels of virus transmission and impact on communities. Thus, a careful analysis of the most and least successful states is of great relevance to a response to the ongoing threat of COVID-19. Moreover, it is worthwhile to compare and contrast the state-by-state behaviours of the pandemic between the three countries as a whole.

In the US, India and Brazil, as well as throughout the world, the scientific response to COVID-19 has been as multifaceted and as significant as the government response. Medical researchers have uncovered numerous means of treating infections \cite{Remdesivir,Bloch2020,toczilizumab,Cao2020}, culminating in the production of vaccines \cite{Polack2020,Walsh2020}. Outside the medical field, analytical approaches to model and study the virus and its impact have been broad. First, many models based on existing mathematical models, such as the Susceptible–Infected–Recovered (SIR) model and the reproductive ratio $R_0$, have been proposed and systematically collated by researchers \cite{Wynants2020,ModellingEstrada2020}. These have been utilised for various purposes, including diagnosis and prognosis of COVID-19 patients, studies of the efficacy of medications, and vaccine development. Next, nonlinear dynamics researchers have proposed several sophisticated extensions to the classical predictive SIR model, including analytic techniques to find explicit solutions \cite{SIRBarlow2020, SIRWeinstein2020}, modifications to the SIR model with additional variables \cite{SIRNg2020,SIRVyasarayani2020,SIRCadoni2020,SIRNeves2020,SIRComunian2020,Sun2020}, incorporation of Hamiltonian dynamics \cite{SIRBallesteros2020} or network models \cite{SIRLiu2021}, and a closer analysis of uncertainty in the SIR equations \cite{Gatto2021}. Other mathematical approaches to prediction and analysis include power-law models \cite{Manchein2020,Blasius2020,Beare2020}, forecasting models \cite{Perc2020}, fractal approaches \cite{Boccaletti2020,Castillo2020,Castillo2021}, neural networks \cite{Melin2020}, Bayesian methods \cite{Manevski2020}, distance analysis \cite{James2021_virulence}, network models \cite{Shang2020, Karaivanov2020,Ge2020,Xue2020}, analyses of the dynamics of transmission and contact \cite{Saldaa2020,Danchin2021},  clustering \cite{Machado2020,jamescovideu} and many others \cite{Ngonghala2020,Cavataio2021,james2021_mobility,Nraigh2020,Glass2020}. Finally, numerous articles have been devoted to understanding the spatial components of the virus' spread, in numerous countries \cite{Zhou2020_covidUS,Melin2020_2,Wang2020_spatioUS,James2021_geodesicWasserstein}.

We have a different motivation and approach relative to the aforementioned work. Numerous works have studied trends in COVID-19 prevalence on a country-by-country basis \cite{Jamesfincovid} or state-by-state basis, frequently within the US or Brazil \cite{daSilva2021,james2020covidusa}. However, we are unaware of any work to consider more than one federation of states at once. We were motivated to compare the US, India and Brazil for several reasons. First, these are the three countries most impacted by COVID-19, both in case and death counts. Secondly, the level of human development varies drastically from country to country, but less so within each federation of states. Third, during the COVID-19 pandemic, international movement drastically decreased, leaving such large federal states almost as self-contained regions in which COVID-19 spread independently from what was occurring in other countries. Thus, tracking the heterogeneity of COVID-19 prevalence and behaviour within and between federations could be used to distinguish the effects of policies at the state and federal level. For example, countries whose federal government had less of a policy role could see more heterogeneity of behaviours with states, if states implemented drastically differing policies.


This work could assist various researchers in different fields. Analysing and predicting the spread of COVID-19 is consistently challenging due to the inability to establish true control groups; indeed, it is practically impossible to split entire countries into different regions where certain mitigation measures are or are not implemented. By comparing states within and between different federations, policy researchers can approximate the existence of control groups, and investigate which socioeconomic features and interventions were associated with better and worse outcomes. For policymakers, a comparison of different states within each federation can provide opportunities for state governments to learn from each other's triumphs and setbacks. Across the three countries, this analysis could reveal relationships between COVID-19 spread and the intervention of the national government or the underlying level of economic development.

This paper is structured in such a way as to thoroughly investigate numerous aspects of the spread and human cost of COVID-19 in the three federations. First, Section \ref{sec:Trajectories} investigates the structural similarity and anomalies in the trajectories of cases, deaths and rolling mortality rate on a state-by-state basis in the three countries. We explore commonalities in virus behaviour within the three countries as well as the extent of heterogeneity across each country as a whole. Next, Section \ref{sec:Turningpoints} performs a closer analysis of a highly significant aspect of COVID-19 epidemiology: differing waves of the outbreak. Using a newly introduced turning point algorithm and distance between finite sets, we perform clustering on all the individual states of the US, India and Brazil to identify characteristic wave behaviours across the entire collection. Finally, Section \ref{sec:Offsets} draws upon the previous two sections to address a highly pertinent metric - the average progression between cases and deaths. This paper introduces a variety of novel optimisation methods to estimate this, and takes a new approach, separating this feature according to the mathematically determined waves of the pandemic. We employ five different optimization methods, each of which uses state-by-state data, to estimate an appropriate offset between case and death time series for the US, India and Brazil as a whole. This allows us to track the changing nature of COVID-19 mortality among the different waves of the pandemic. We summarise all our findings and insights in Section \ref{sec:conclusion}.

In addition to the above motivation and specific questions we study, the methodologies used in this paper have applicability well beyond the COVID-19 pandemic, and could be used in any setting of multivariate time series. In particular, Section \ref{sec:Trajectories} presents a new approach to carefully quantify the extent of heterogeneity in a multivariate time series (or in other spaces more generally) that handle the existence of outlier elements well, while Section \ref{sec:Offsets} could be used to study various other time series where lagging is to be expected. Given the fourth wave of COVID-19 that Europe is currently facing, scientists should seek to learn from the countries most severely impacted by COVID-19, and their prior waves of COVID-19 cases. This manuscript provides computational tools and findings that would be of great relevance to this audience.

\section{Trajectory analysis, structural similarity and anomaly detection}
\label{sec:Trajectories}

In this section, we explore the similarity and structure between case, death and rolling mortality time series for the US, India and Brazil. Our data spans 26 Feb 2020 to 23 May 2021, a period of $T=452$ days. For each country, let the multivariate time series of new COVID-19 cases and deaths be $x_i(t)$ and $y_i(t)$, where $t=1,...,T$ indexes the days and $i = 1,...,N$ indexes states under consideration. Throughout this manuscript, we will examine either one country at a time, with $N=51$ states (including the District of Columbia) for the US, $N=36$ states (including union territories) for India, $N=27$ states (including the Federal District) for Brazil, or the entire collection of individual states, with $N=114$.

In addition, we define a 30-day rolling mortality rate for each state as follows:
\begin{equation}
    r_i(t) = \frac{\sum^{t}_{j=t-29} y_i(j)}{\sum^{t}_{j=t-29} x_i(j)}, \quad t=30,...,T.
\end{equation}
We wish to examine the three aforementioned multivariate time series to determine the structure and degree of heterogeneity within each country's states and collectively, between all countries' underlying states. To a case time series $x_i(t), t=1,...,T$ we associate the following probability distribution:
\begin{align}
\label{eq:associateddistribution}
    f_i=\frac{1}{\sum_{s=1}^T x_i(s) } \sum_{t=1}^T x_i(t) \delta_t,
\end{align}
where $\delta_t$ is the Dirac delta distribution at $t$. That is, $f_i$ is a distribution that apportions to day $t$ the weight of the new cases observed on that day as a proportion of the total cases across the whole period. Then, we define
\begin{align}
    M^{C}_{ij}=W_1(f_i,f_j),
\end{align}
where $W_1$ is the $L^1$-Wasserstein metric \cite{DelBarrio} between distributions on $\mathbb{R}$. Analogously, we associate distributions $g_i$ and $h_i$ to death and mortality time series $y_i$ and $r_i$, respectively. We define \emph{trajectory distance matrices} between  state trajectories for deaths and mortality analogously as follows:
\begin{align}
    M^D_{ij} &= W_1(g_i,g_j); \\
    M^R_{ij} &= W_1(h_i,h_j).
\end{align}
This distance has several advantageous properties over previously used discrepancy measures between normalised trajectories. Previous work \cite{james2020covidusa} has used the $L^1$ norm and metric between normalised trajectories, defined as follows:
\begin{align}
\|x_i\|_1= \sum_{t=1}^T x_i(t) \\
\mathbf{v}_i=\frac{x_i}{\|x_i \|_1}\\
  d_{ij}=  \|\mathbf{v}_i - \mathbf{v}_j\|_1
\end{align}
This treats each time series $x_i(t)$ as a vector in $\mathbb{R}^T$, normalises by its $L^1$ norm, and compares these normalised vectors with the $L^1$ metric \cite{Minkowski}. This distance is suitable in most instances but has some undesirable properties when quantifying discrepancy between noisy time series. Specifically, this $L^1$ distance $d_{ij}$ has maximal possible value equal to 2 when $x_i(t)$ and $x_j(t)$ have disjoint support. Practically, this would mean that two states' trajectories would receive a large $L^1$ discrepancy measure if the cases were simply reported to fall on different days. For example, if state $i$ and state $j$ had broadly similar trends in cases, but in state $i$ cases were reported more on Mondays and Wednesdays while state $j$ reported more on Tuesdays and Thursdays, then the $L^1$ distance measure would be larger than their similarity. Smoothing and 7-day averaging can resolve some of these issues, but the Wasserstein metric ameliorates this issue even more, as it is robust to small translations of distributions. That is, if $f$ is a distribution and $f_\delta(x)=f(x+\delta)$, then $W_1(f,f_\delta)=|\delta|$, as shown in \cite{James2020_nsm}. This means the Wasserstein metric assigns a low value in the case that states $i$ and $j$ have similar trajectories where cases just fall on nearby but distinct days.

We will examine the matrices defined above ($M^C, M^D$ and $M^R$) for each individual country (with $N=51$ for the US, 36 for India, 27 for Brazil) as well as the entire collection of states, with $N=114$. In Figure \ref{fig:TrajectoryDistanceMatrices}, we display the matrices $M^C, M^D$, and $M^R$ each for the totality of the collection. In Table \ref{tab:Trajectory_norm_table}, we record the $L^1$-norms $\| M^C \|,\| M^D \|$, and  $\| M^R \|$ each restricted to one of the three federations. For example, for the US, $M^C, M^D$ and  $M^R$ are $51 \times 51$ matrices, whereas they are $36 \times 36$ matrices for India. For an $N\times N$ matrix $A$, we define its norm by
\begin{align}
\label{eq:matrixnorm}
    \|A\|=\frac{1}{N(N-1)} \sum_{i,j=1}^{N} |A_{ij}|.
\end{align}
This calculates a total magnitude of the matrix, appropriately normalised for the number of non-zero elements. For our distance matrices $M^C, M^D$, and $M^R$, these norms reflect the heterogeneity among trajectories within each country. As the Wasserstein distance is taken between appropriately normalised distributions, it is possible to compare between case, death and mortality time series. Due to the normalisation coefficient, it is possible to compare this between different countries.

Table \ref{tab:Trajectory_norm_table} reveals that India exhibits the highest heterogeneity between states regarding all three behaviours, with norms of 40.09, 55.08 and 76.45 for cases, deaths and mortality, respectively. For case trajectories, the US and Brazil have similar levels of total homogeneity. For deaths and mortality trajectories, however, Brazil's norms of 35.33 and 29.67 are rather less than the US' scores of 43.70 and 43.03. This highlights the relative homogeneity in death and mortality trajectories among Brazilian states. 

\begin{figure*}
    \centering
    \begin{subfigure}[b]{0.49\textwidth}
        \includegraphics[width=\textwidth]{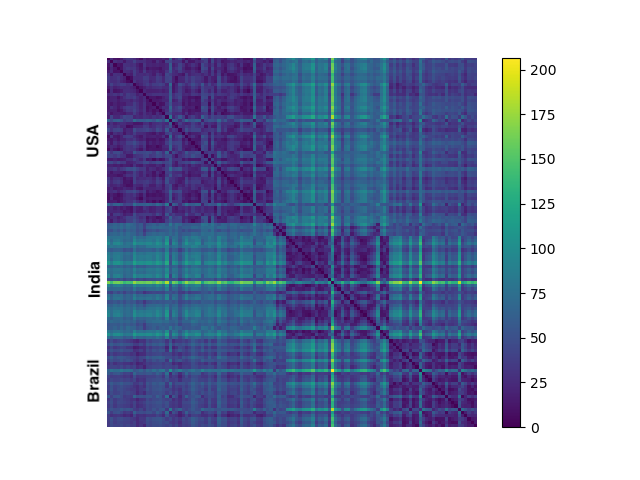}
        \caption{}
        \label{fig:Cases_trajectory}
    \end{subfigure}
    \begin{subfigure}[b]{0.49\textwidth}
        \includegraphics[width=\textwidth]{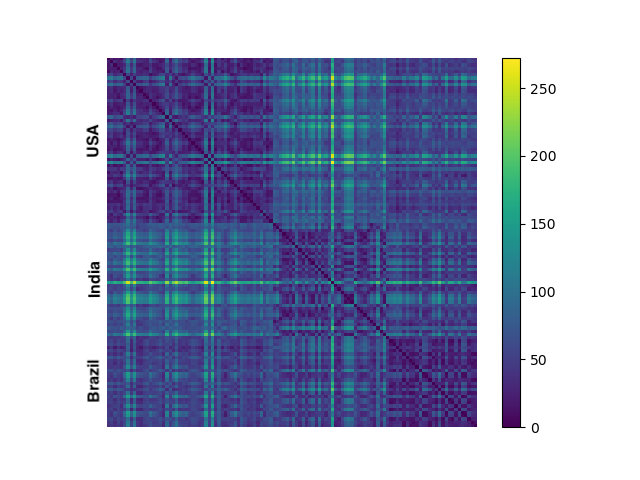}
        \caption{}
        \label{fig:Deaths_trajectory}
    \end{subfigure}
    \begin{subfigure}[b]{0.49\textwidth}
        \includegraphics[width=\textwidth]{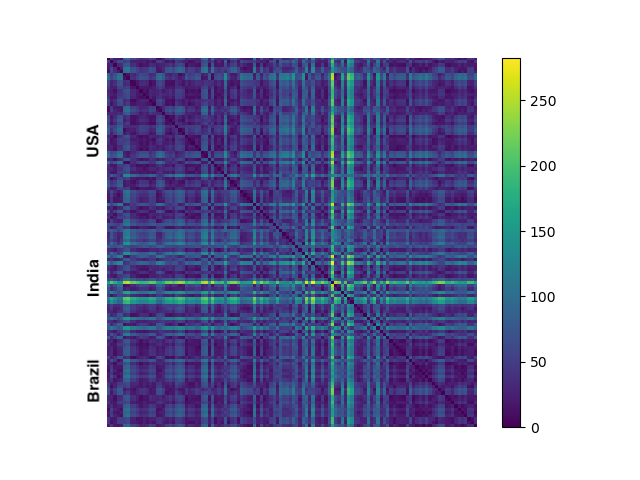}
        \caption{}
        \label{fig:Mortality_trajectory}
    \end{subfigure}
    \caption{Trajectory distance matrices, as defined in Section \ref{sec:Trajectories}, with respect to (a) cases, (b) deaths and (c) mortality rate time series. Each matrix is computed using the entire collection of 114 states and ordered with US states first, then Indian states, then Brazil. Darker values indicate smaller entries of the matrix, signifying greater similarity between states. India exhibits particular heterogeneity between mortality rates.}
    \label{fig:TrajectoryDistanceMatrices}
\end{figure*}

\begin{table}
\begin{center}
\begin{tabular}{ |p{2cm}||p{1.5cm}|p{1,5cm}|p{2.5cm}|}
 \hline
 \multicolumn{4}{|c|}{Trajectory distance matrix norms} \\
 \hline
 Country & Cases & Deaths & Mortality rate \\
 \hline
 US & 27.39 & 43.70 & 43.03   \\
 India & 40.09 & 55.08 & 76.45   \\
 Brazil & 30.30 & 35.33 & 29.67 \\  
\hline
\end{tabular}
\caption{Normalised matrix norms $\|M^C\|, \|M^D\|$, and  $\|M^R\|$, as defined in Section \ref{sec:Trajectories}, for each of the three countries with respect to case, death and mortality time series. The higher values for India indicate greater heterogeneity between its states.}
\label{tab:Trajectory_norm_table}
\end{center}
\end{table}

Next, we wish to further examine the heterogeneity between states of each country, as well as identify the presence of any outlier states that may be influencing the total norms recorded in Table \ref{tab:Trajectory_norm_table}. Given each country's trajectory matrix $M$ (with respect to cases, deaths or mortality rates), we perform the following procedure to sequentially identify the most anomalous state, remove it, and compute the resulting norm of the reduced collection. This is described in Algorithm \ref{algorithm}.

\begin{algorithm}[H]
	\caption{Anomalous trajectory identification}
	\label{algorithm}
	\begin{algorithmic}[1]		
		\State Input: a distance matrix $M$ between all states for a candidate country.
		\State Initialise empty lists for norm scores and state anomaly ranking.
		\State Set matrix $M$ to current iteration, $M^c$.
		\For{$r=0 \,\text{ to }\, N-1$}
		\State Compute number of rows, $m^c = M^c - r$, in current distance matrix, $M^c$.
        \State Generate current norm of matrix $M^c$, $\nu^c = \frac{1}{m^c (m^c-1)} \sum_{i,j=1}^{m^c} |M^c_{ij}|$.
        \State Append $\nu^c$ to norm scores list.
		\For{$l=1 \,\text{ to }\, m^c$}
		\State Compute $a_l = \sum_{i=1}^{m^c} |M^c_{il}|$.
		\EndFor 
		\State Let $a_k=\max_l \{ a_l \}$
		\State Append name of state $k$ to state anomaly ranking list.
		\State Update matrix $M^c$ by removing $k$th column and row corresponding to state with highest anomaly score, $a_k$.
		\EndFor 
		\State Generate norm scores trajectory and state anomaly ranking.
	\end{algorithmic}
\end{algorithm}

\begin{figure*}
    \centering
    \begin{subfigure}[b]{0.325\textwidth}
        \includegraphics[width=\textwidth]{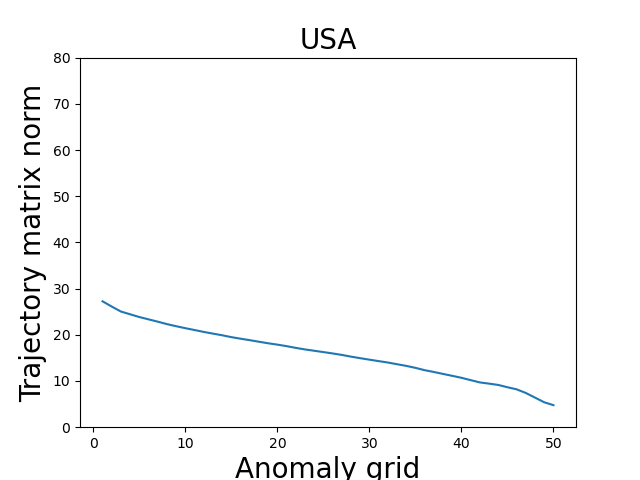}
        \caption{}
        \label{fig:USA_sequential_cases}
    \end{subfigure}
    \begin{subfigure}[b]{0.325\textwidth}
        \includegraphics[width=\textwidth]{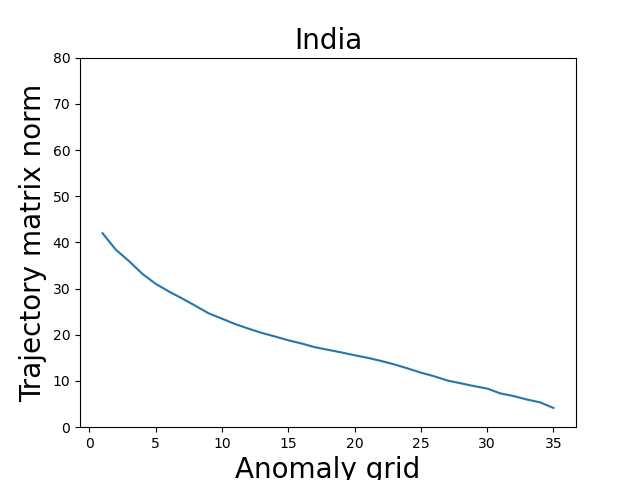}
        \caption{}
        \label{fig:India_sequential_cases}
    \end{subfigure}
    \begin{subfigure}[b]{0.325\textwidth}
        \includegraphics[width=\textwidth]{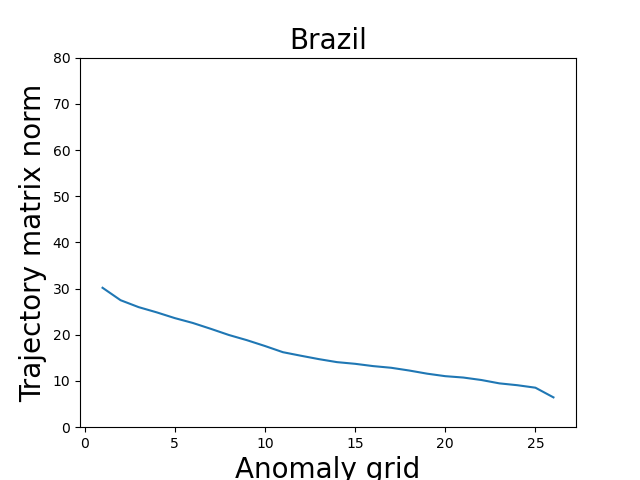}
        \caption{}
        \label{fig:Brazil_sequential_cases}
    \end{subfigure}
    \begin{subfigure}[b]{0.325\textwidth}
        \includegraphics[width=\textwidth]{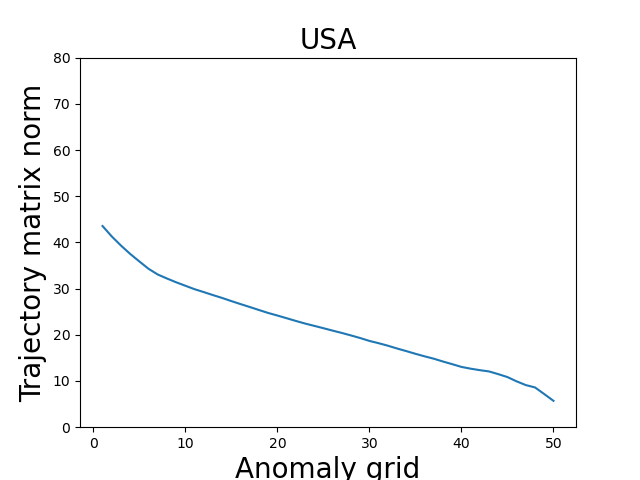}
        \caption{}
        \label{fig:USA_sequential_deaths}
    \end{subfigure}
    \begin{subfigure}[b]{0.325\textwidth}
        \includegraphics[width=\textwidth]{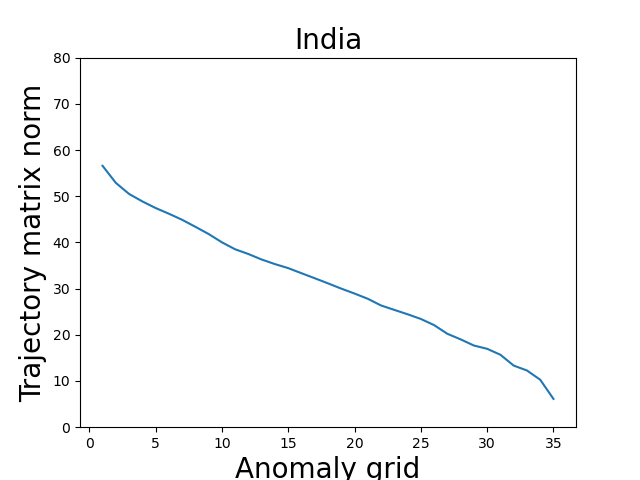}
        \caption{}
        \label{fig:India_sequential_deaths}
    \end{subfigure}
    \begin{subfigure}[b]{0.325\textwidth}
        \includegraphics[width=\textwidth]{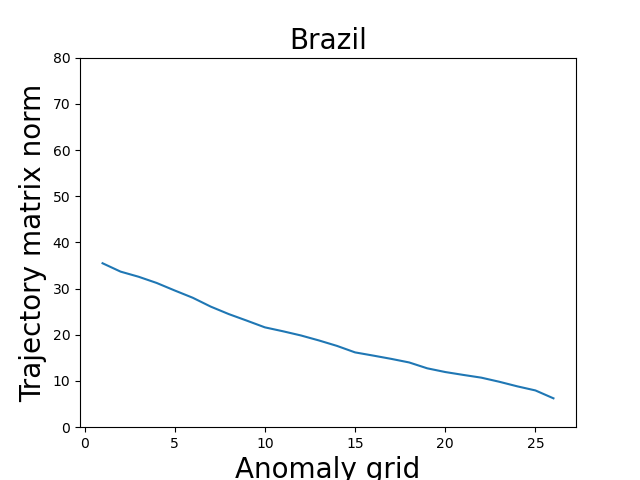}
        \caption{}
        \label{fig:Brazil_sequential_deaths}
    \end{subfigure}
    \begin{subfigure}[b]{0.325\textwidth}
        \includegraphics[width=\textwidth]{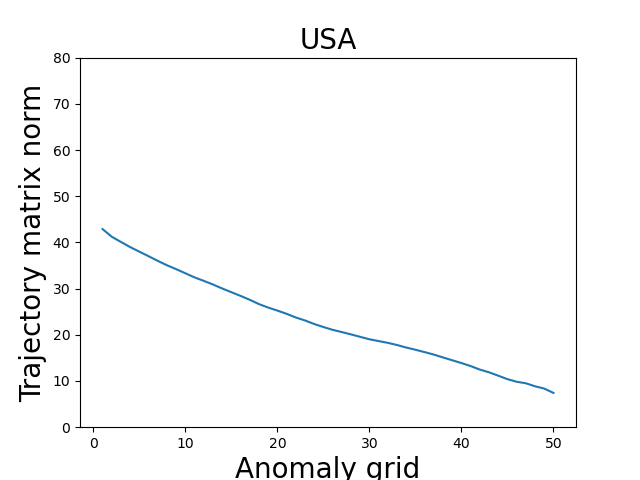}
        \caption{}
        \label{fig:USA_sequential_rates}
    \end{subfigure}
    \begin{subfigure}[b]{0.325\textwidth}
        \includegraphics[width=\textwidth]{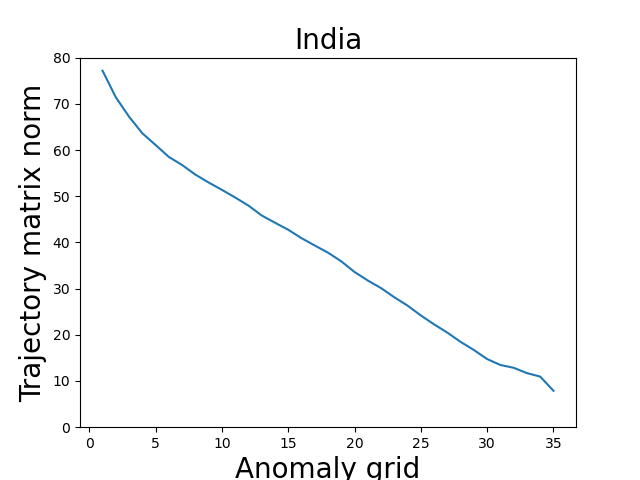}
        \caption{}
        \label{fig:India_sequential_rates}
    \end{subfigure}
    \begin{subfigure}[b]{0.325\textwidth}
        \includegraphics[width=\textwidth]{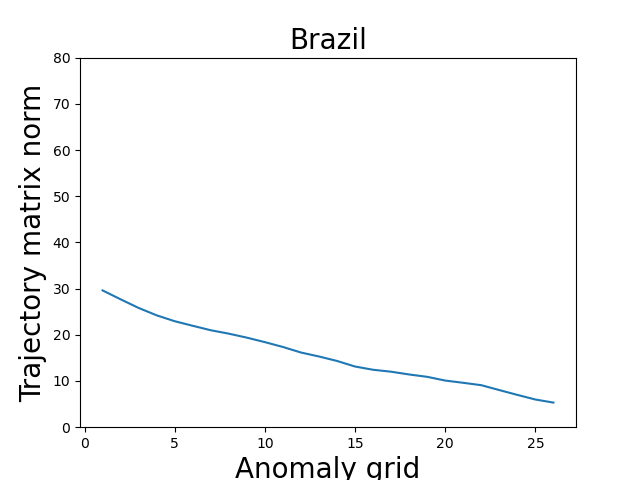}
        \caption{}
        \label{fig:Brazil_sequential_rates}
    \end{subfigure}
    \caption{Sequences of decreasing matrix norms $\nu^c$, as determined in Algorithm \ref{algorithm}, for (a) US cases (b) Indian cases (c) Brazilian cases (d) US deaths (e) Indian deaths (f) Brazilian deaths (g) US mortality rates (h) Indian mortality (i) Brazilian mortality. These norms are obtained sequentially after removing the most anomalous country at each step. India exhibits sharper drops at the start, indicating a small collection of highly anomalous states, particularly for mortality rates.}
    \label{fig:Sequential_anomaly_removal_cases}
\end{figure*}

In Figure \ref{fig:Sequential_anomaly_removal_cases}, we display the sequence of norm scores $\nu^c$ for each matrix $M^C, M^D$, and $M^R$ for the US, India and Brazil. By removing the greatest $a_k$ in each step of the algorithm, this sequence of norm scores is necessarily decreasing. As all norms are appropriately normalised, we may compare these decreasing sequences between all our different countries and time series. Several insights can be gained from these figures. First, India consistently produces the largest anomaly score for all three attributes. This can be seen by the magnitude of the decreasing trend for India throughout the plots. This is consistent with the analysis in Table \ref{tab:Trajectory_norm_table}, but ensures that it is not due simply to the presence of a small number of outlier states. Second, relative to cases and deaths, mortality rate trajectories are significantly more dissimilar in the case of India. For the US and Brazil, there is greater uniformity in anomaly trajectories among each of the three attributes. When examining the nine sequential norm trajectories, it is pertinent to look for sharp drops, which would indicate that a particular state accounts for a disproportionate amount of heterogeneity. This effect is seen in the Indian mortality rate norms (Figure \ref{fig:India_sequential_deaths}) and to a lesser extent in the cases and deaths norms, (Figures \ref{fig:India_sequential_cases} and Figure \ref{fig:India_sequential_deaths}, respectively).

Table \ref{tab:State_anomaly_table} records the five most anomalous states in each country with respect to cases, deaths and mortality rates, as determined by Algorithm \ref{algorithm}, and also reveals several insights. In the US, there is a pronounced geographic trend in all three attributes' anomaly trajectories. Northeastern states New York, New Jersey, Connecticut and Vermont are identified as anomalous in at least two attributes' trajectories each. Several other Northeastern states appear, such as New Hampshire, Maine, Massachusetts and DC. In addition, there is substantial consistency in the states exhibiting anomalous behaviours in cases, deaths and mortality. In India, the state Lakshadweep is the most anomalous in cases, deaths and mortality, but otherwise relatively less repetition is observed among the most anomalous states. Lakshadweep's status as an anomaly can also explain the sharp drops observed for India in Figure \ref{fig:Sequential_anomaly_removal_cases}, but not for the US or Brazil. Brazil exhibits even greater variability in the most anomalous states than the US or India, with little consistency in the states exhibiting anomalous behaviours among cases, deaths and mortality.

\begin{table}
\begin{center}
\begin{tabular}{ |p{1.3cm}|p{4.2cm}|p{3.8cm}|p{3cm}|}
 \hline
 Country & Cases & Deaths & Mortality \\
 \hline
 US$_1$ & Vermont & New York & Oklahoma \\
 US$_2$ & Maine & New Jersey & Vermont \\
 US$_3$ & New Hampshire & Connecticut & New Jersey \\
 US$_4$ & New York & DC & Connecticut \\
 US$_5$ & Michigan & Massachusetts & New York \\
 India$_1$ & Lakshadweep & Lakshadweep & Lakshadweep \\
 India$_2$ & Andaman \& Nicobar Islands & Tripura & Mizoram \\
 India$_3$ & Tripura & Andhra Pradesh & Nagaland \\
 India$_4$ & Arunachal Pradesh & Odisha & Himachal Pradesh \\
 India$_5$ & Assam & Dadra and Nagar Haveli & Gujarat \\
 Brazil$_1$ & Maranh{\~a}o & Pernambuco & Pernambuco \\
 Brazil$_2$ & Roraima & Paran{\'a} & Piau{\'i} \\
 Brazil$_3$ & Amap{\'a} & Minas Gerais & Cear{\'a} \\
 Brazil$_4$ & Distrito Federal & Rio Grande do Sul & Distrito Federal \\
 Brazil$_5$ & Minas Gerais & Santa Catarina & Para{\'i}ba \\
\hline
\end{tabular}
\caption{The five most anomalous states in each country with respect to case, death and mortality rate time series, as determined by Algorithm \ref{algorithm}.}
\label{tab:State_anomaly_table}
\end{center}
\end{table}

\section{Wave behaviour analysis}
\label{sec:Turningpoints}

In this section, we investigate one of the most significant aspects of the spread of COVID-19, the tendency for the virus to exhibit multiple distinct waves of prevalence. As in the last section, we analyse either each country on a state-by-state basis (with $N=51, 36,$ and  $27$ states) or the entire collection of states across the three countries together ($N=114$ states). To each state, we apply a newly introduced turning point algorithm \cite{james2020covidusa} to identify non-trivial local maxima (peaks) and minima (troughs) in the new case time series. 

We first apply a \emph{Savitzky-Golay filter} to each new case time series $x_i(t)$ to generate a smoothed collection of time series $\hat{x}_i(t)$, $t=1,...,T$ and $i=1,...,N$. We then apply a two-stage turning point algorithm, detailed in \ref{appendix:turningpoint}, to generate non-empty sets $P_i$ and $T_i$ of non-trivial local maxima (peaks) and local minima (troughs), respectively. These turning points alternate between a trough and peak, beginning with a trough at $t=1$, when there are no cases.  

Next, we use an appropriate distance measure to quantify the similarity between two sets of turning points. We apply the semi-metric first introduced in \cite{James2020_nsm}. Given two non-empty finite sets $A,B$, this is defined as
\begin{align}
\label{eq:MJ}
    D({A},{B}) = \frac{1}{2} \left(\frac{\sum_{b\in B} d(b,A)}{|B|} + \frac{\sum_{a \in {A}} d(a,B)}{|A|} \right),
\end{align}
where $d(b,A)$ is the minimal distance from $b \in B$ to the set $A$. The distance measure $D(A,B)$ is symmetric, non-negative, and zero if and only if $A=B$. We then define $N \times N$ turning point distance matrices $M^{TP}$ by
\begin{align}
\label{eq:DTPmatrix}
    M_{ij}^{TP} = D(P_i,P_j) + D(T_i,T_j).
\end{align}
As before, this may be computed for the entire collection ($N=114$) or one specific country. In Figures \ref{fig:USA_TP_dendrogram}, \ref{fig:India_TP_dendrogram} and \ref{fig:Brazil_TP_dendrogram}, respectively, we display hierarchical clustering on the three obtained turning point matrices $M^{TP}$ restricted to the states of the US, India and Brazil separately.

Examining these three dendrograms reveals a similar cluster structure between the US and India. Both countries display a dense majority cluster and a small collection of outlier states. Brazil, by contrast, exhibits quite a different structure, with two similarly sized clusters that contain the majority of elements, and then some outliers. We can further examine the cluster-split behaviour of Brazil by examining the results of clustering all $N=114$ states in our collection in Figure \ref{fig:TP_MJ_Dend}. This total dendrogram contains a majority cluster containing $\sim$ 90\% of all states, and two small outlier clusters of five and four states (clusters B and C respectively). The majority cluster contains two subclusters (A1 and A2), featuring a break between US and Indian states, with almost no intersection between the two countries. However, Brazil's states are far more widely distributed. Not only do the outlier clusters B and C consist only of Brazilian states, but Brazil's states are spread throughout both A1 and A2, interleaving between US and Indian states. This finding suggests that US and Indian states exhibit higher intra-collection homogeneity and inter-collection heterogeneity in their wave behaviours when compared to Brazilian states.

To elucidate the reasons behind these state clustering patterns, we study the distribution of the location of the first non-trivial trough, $T_1$. This trough indicates the end of the first wave; thus, the value $T_1$ gives the total length of the first wave in each state. Table \ref{tab:TP1_table} documents the median and standard deviation of $T_1$ among each country's states, while Figure \ref{fig:Distribution_TP_1} displays kernel density estimates of the full distribution of values. There is significant variability between the states' first wave lengths between the three countries. The US has a median value of 92 and a standard deviation of 76.9, indicating that most states experienced a short first wave. By contrast, Indian states mostly experienced a long first wave, with a median value of 231 and a standard deviation of 63.6. This suggests that the first wave of COVID-19 cases in Indian states was on average 2.5 times longer than US states, with limited variance between states. As in Figure \ref{fig:TurningPointDendrograms}, Brazil does not exhibit as strong a characteristic behaviour, with a median $T_1$ score of 143 and a significantly higher standard deviation among Brazilian states of 109. Notably, the median $T_1$ value of Brazilian states is located between the US and Indian median values. Also of note is the highly skewed distribution for the Brazilian states, with a substantial number of high values despite the relatively lower peak. When viewed in conjunction with Figure \ref{fig:TP_MJ_Dend}, one can see how the heterogeneous turning point behaviours of Brazilian states are classified into predominantly US or Indian subclusters (A1 and A2, respectively). Figure \ref{fig:Distribution_TP_1} shows in more detail that the lengths of the first wave among Brazilian states are broadly positioned between those of US and Indian states.

\begin{figure*}
    \centering
    \begin{subfigure}[b]{0.7\textwidth}
        \includegraphics[width=\textwidth]{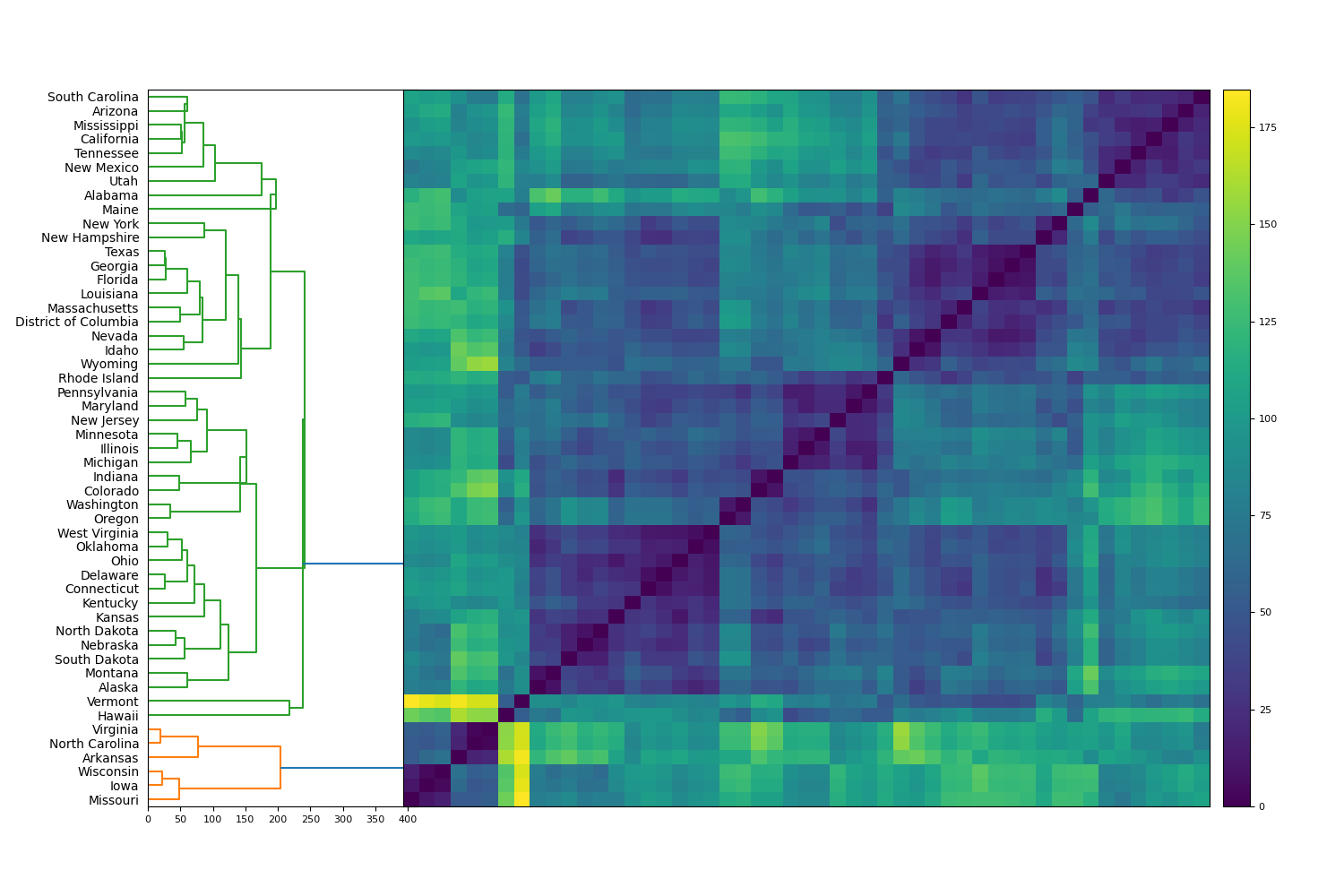}
        \caption{}
        \label{fig:USA_TP_dendrogram}
    \end{subfigure}
    \begin{subfigure}[b]{0.7\textwidth}
        \includegraphics[width=\textwidth]{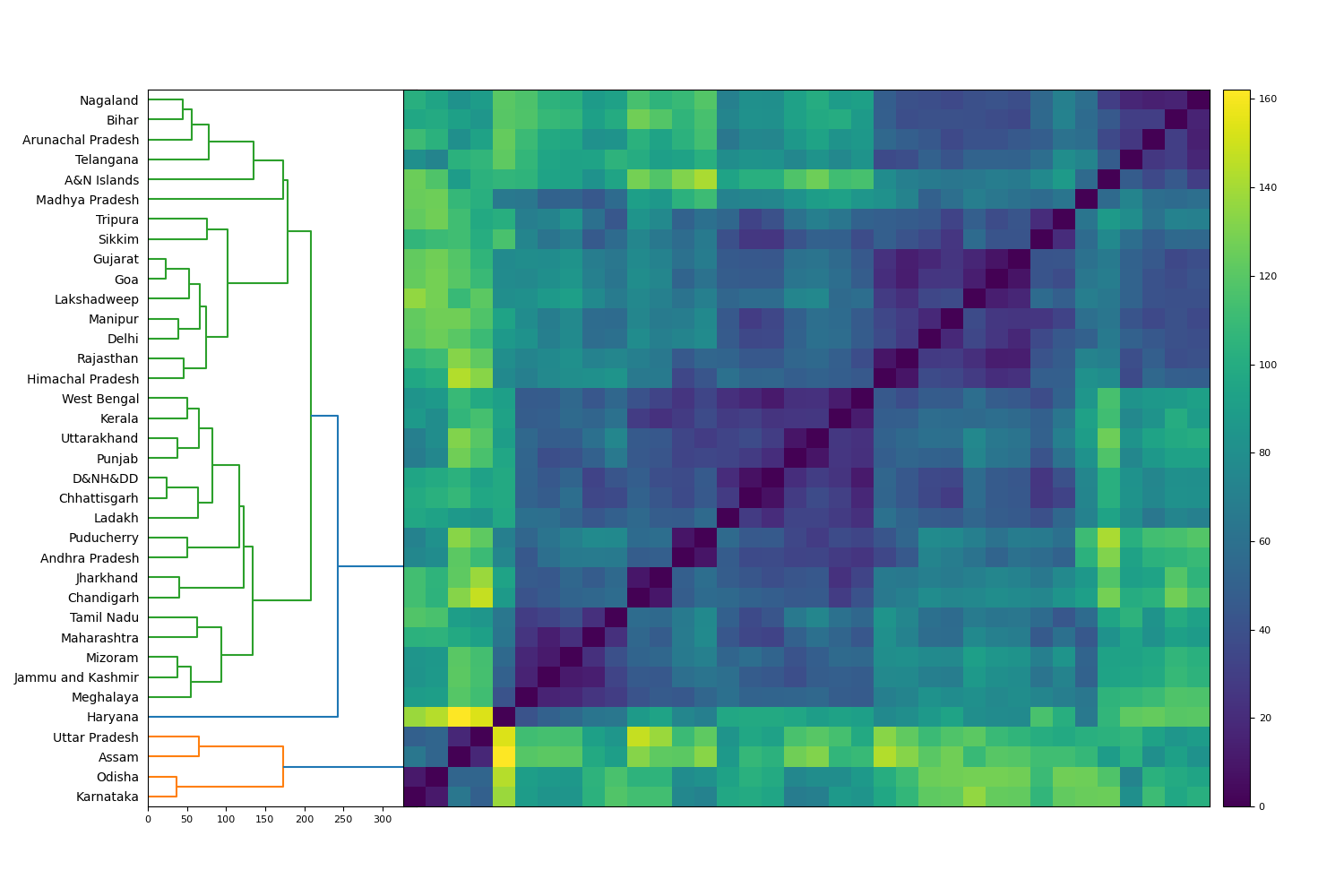}
        \caption{}
        \label{fig:India_TP_dendrogram}
    \end{subfigure}
    \begin{subfigure}[b]{0.7\textwidth}
        \includegraphics[width=\textwidth]{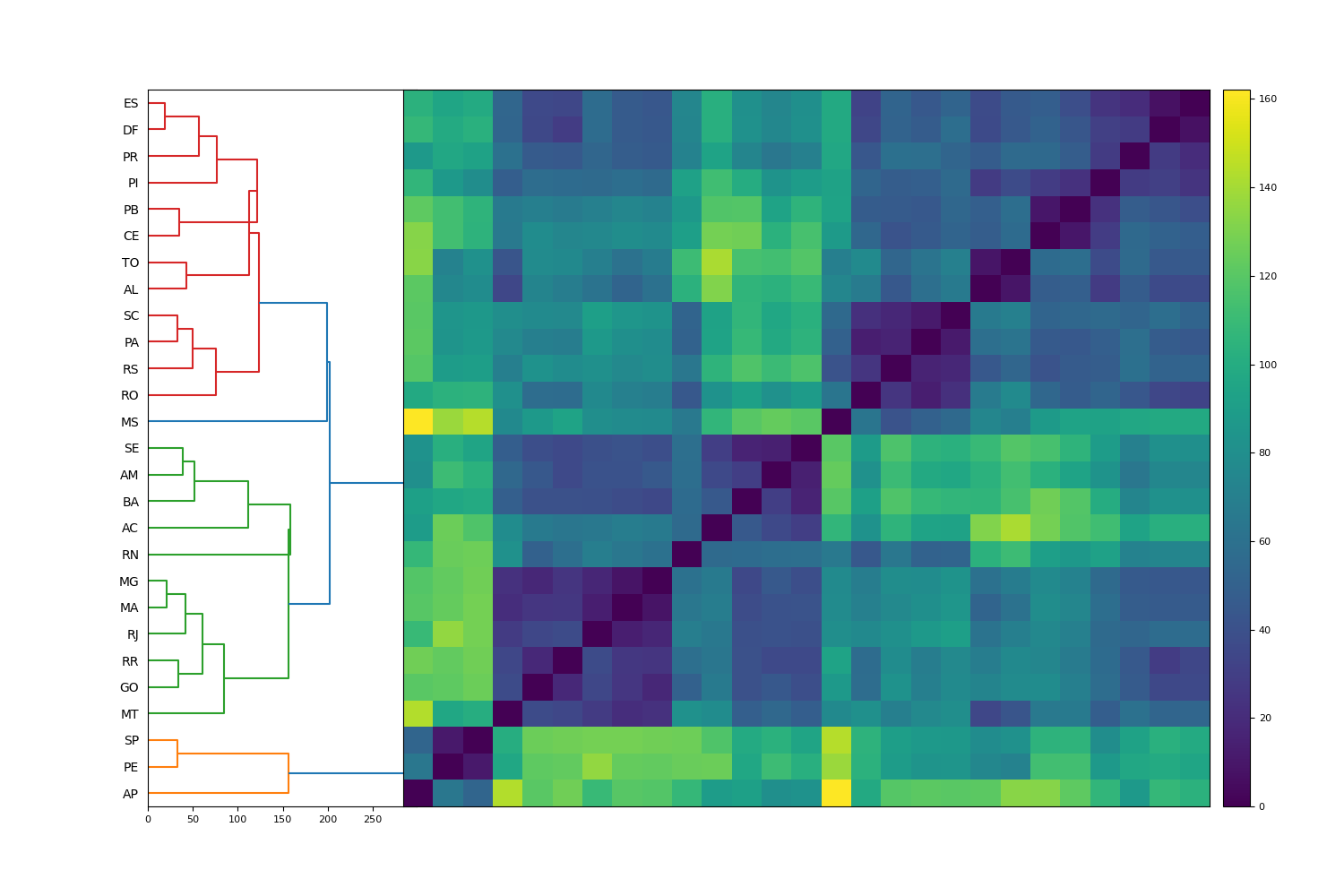}
        \caption{}
        \label{fig:Brazil_TP_dendrogram}
    \end{subfigure}
    \caption{Hierarchical clustering on the matrix $M^{TP}$, defined in Section \ref{sec:Turningpoints}, between the individual sets of (a) US states (b) Indian states (c) Brazilian states. A broadly similar cluster structure is observed between the US and India, while Brazil's structure is quite different.}
    \label{fig:TurningPointDendrograms}
\end{figure*}

\begin{figure}
    \centering
    \includegraphics[width=1.0\textwidth]{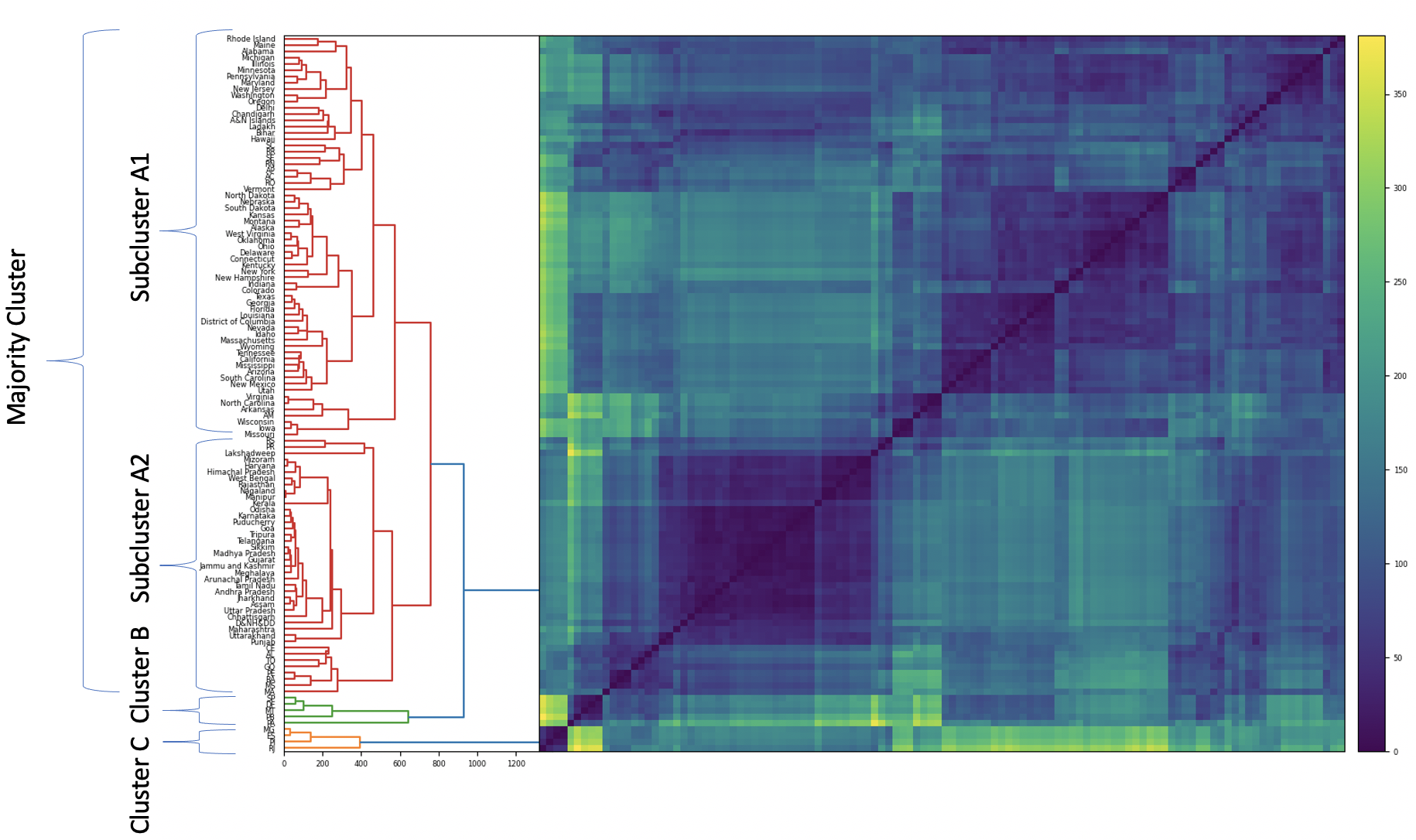}
    \caption{Hierarchical clustering on the matrix $M^{TP}$, defined in Section \ref{sec:Turningpoints}, for all $N=114$ states in our collection. The majority cluster contains two subclusters A1 and A2, broadly consisting of US states and Indian states, respectively. Brazilian states, labelled with two letters for visibility, are interleaved among A1 and A2 and also the two outlier clusters B and C.}
    \label{fig:TP_MJ_Dend}
\end{figure}

\begin{table}
\begin{center}
\begin{tabular}{ |p{1.3cm}|p{2.2cm}|p{3.5cm}|}
 \hline
 Country & Median $T_1$ & Standard deviation $T_1$ \\
 \hline
 US & 92 & 76.9 \\
 India & 231 & 63.6 \\
 Brazil & 143 & 109 \\
\hline
\end{tabular}
\caption{Median and standard deviation of the length of the first wave $T_1$, defined in Section \ref{sec:Turningpoints} and measured by the first non-trivial trough. The US has the shortest first wave, India has the longest, while Brazil exhibits the greatest variability.}
\label{tab:TP1_table}
\end{center}
\end{table}

\begin{figure}
    \centering
    \includegraphics[width=0.75\textwidth]{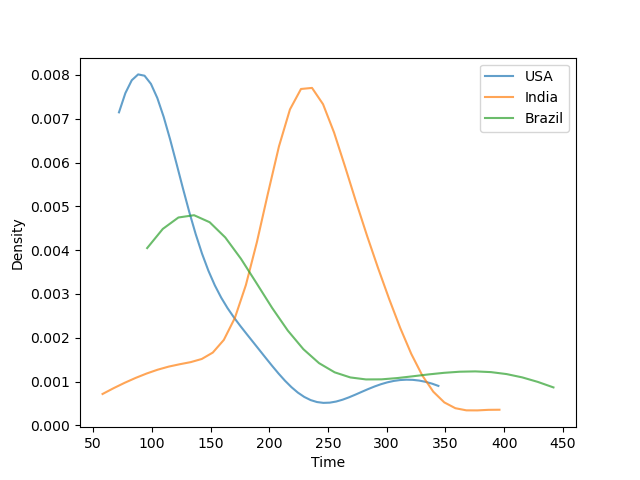}
    \caption{Kernel density estimates of distributions of the first wave length $T_1$, defined in Section \ref{sec:Turningpoints}, over each country. The US exhibits the smallest first wave length, India the greatest, while Brazil has the greatest variability.}
    \label{fig:Distribution_TP_1}
\end{figure}

\section{Offsets between cases and deaths}
\label{sec:Offsets}

In this section, we combine the motivating questions from the previous two sections: the different wave behaviour of the virus, and the time-varying properties of cases, deaths and mortality rates by states. Here, we investigate various methods to quantify and analyse the changing offset between cases to deaths in the different waves of the pandemic in the three countries under consideration. To standardise our comparison of offsets between constituent states, we consider a uniform partition into waves for each entire country. That is, let $x(t)$ be the new daily case time series for an entire country (total counts for the US, India, or Brazil). As in the previous section, we use the methodology of \cite{james2020covidusa}, detailed in \ref{appendix:turningpoint}, to divide each aggregated country's case time series into a first, second and possibly third wave. Let $T_0=1$, $T_1$ be the first non-trivial trough, and $T_2$ be the second non-trivial trough, if it exists. For India and Brazil, this does not exist, so we set $T_2=T$. Then the interval $T_0\leq t \leq T_1$ represents the first wave, $T_1\leq t \leq T_2$ the second wave, and in the case of the US only, $T_2\leq t \leq T$ represents the third wave. For notational convenience, we set $T_3=T$ for the US. Thus, the $k$th wave can be described by the interval $[T_{k-1}, T_k]$, where $k=1,2$ for India and Brazil and $k=1,2,3$ for the US. These turning points for the three country's aggregated cases are displayed in Figure \ref{fig:Total_cases}.

\begin{figure*}
    \centering
    \begin{subfigure}[b]{0.47\textwidth}
        \includegraphics[width=\textwidth]{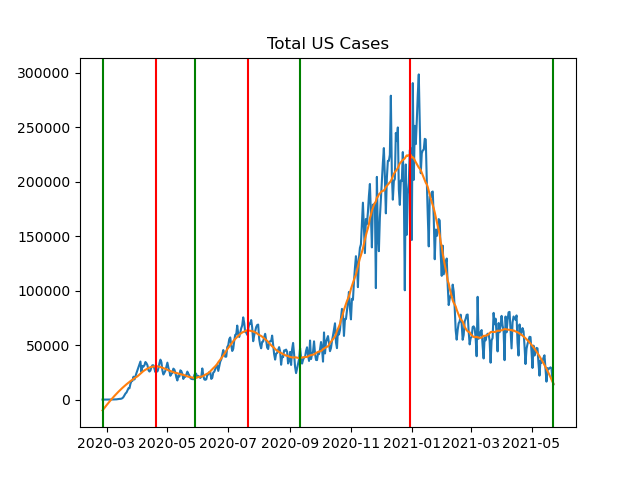}
        \caption{}
        \label{fig:US_cases}
    \end{subfigure}
    \begin{subfigure}[b]{0.47\textwidth}
        \includegraphics[width=\textwidth]{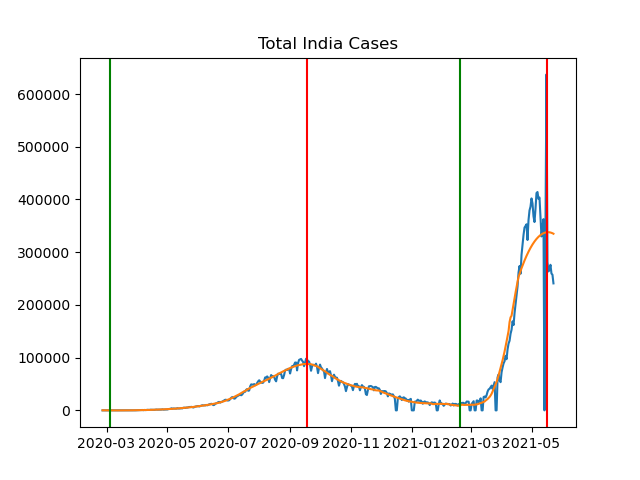}
        \caption{}
        \label{fig:India_cases}
    \end{subfigure}
    \begin{subfigure}[b]{0.47\textwidth}
        \includegraphics[width=\textwidth]{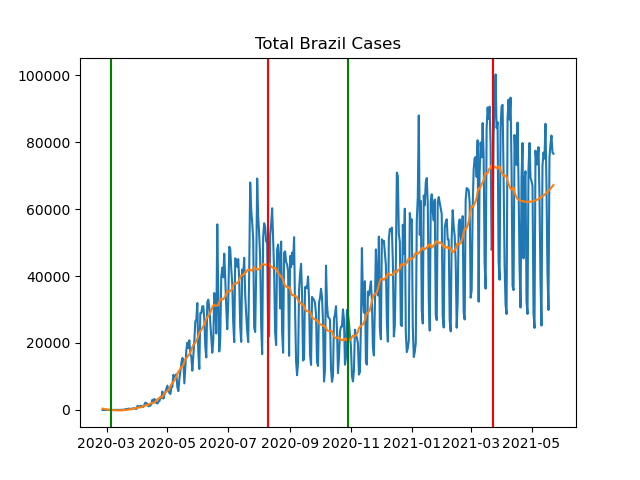}
        \caption{}
        \label{fig:Brazil_cases}
    \end{subfigure}
    \caption{New daily case time series and determined turning points, defined in Section \ref{sec:Turningpoints}, for (a) the US (b) India (c) Brazil.}
    \label{fig:Total_cases}
\end{figure*}

We apply five different methods to estimate suitable values of the offset $\tau_k$ between case and death time series for each wave in each country. Each method determines an appropriate offset using case and death data only between $a=T_{k-1}$ and $b=T_k$. Let $T^* = b - a$ be the length of this interval. We describe the five methods below.

\begin{enumerate}
    \item \textbf{Affinity matrices}: 
For a given wave and country, let the offset $\alpha_k$ be chosen as follows: on each day $t$, let $D_X(t), D_Y(t)$ be the matrices of differences between cases and deaths, respectively. That is, $D_X(t)$ is an $N \times N$ matrix defined by $D_X(t)_{ij}=|x_i(t) - x_j(t)|$, where indices $i,j$ range over the states of one country, and similarly for $D_Y(t)$. To any distance matrix $D$, we can assign a corresponding affinity matrix $A$ defined by
\begin{align}
    A_{ij}=1 - \frac{D_{ij}}{\max D}.
\end{align}
Let $\text{Aff}_X(t)$ and  $\text{Aff}_Y(t)$ be the affinity matrices corresponding to $D_X(t), D_Y(t)$, respectively. Given an offset $\alpha$, with $0 < \alpha < T^*$, let the \emph{normalised total affinity difference} be defined as
\begin{align}
\frac{1}{T^* - \alpha} \sum_{a\leq t \leq b - \alpha} \|\text{Aff}_X(t) - \text{Aff}_Y(t + \alpha)\|.
\label{eq:affinitydiff}
\end{align}
The matrix norm is the same as defined in (\ref{eq:matrixnorm}). Then, the \emph{affinity offset} of a wave is defined as the value $\alpha_k$ that minimises this total difference.

\item \textbf{Probability density function (PDF)}:
For a given wave and country, let the offset $\beta_k$ be chosen as follows: on each day, let $p_X(t), p_Y(t)$ be the probability vector for new cases and deaths on day $t$. That is, $p_X(t)$ is a length $N$ vector defined by $p_X(t)_i = \frac{x_i(t)}{\sum_{j=1}^N x_j(t)} $ where $i$ ranges over the states of one country. Given an offset $\beta$, let the \emph{normalised total pdf difference} be defined as
\begin{align}
\frac{1}{T^*- \beta} \sum_{a\leq t \leq b - \beta} \|p_X(t) - p_Y(t + \beta)\|_1,
\label{eq:pdfdiff}
\end{align}
where $\|\|_1$ is the $L^1$ norm between vectors. Then the \emph{pdf offset} of a wave is defined as the value $\beta_k$ that minimises this total difference.

\item \textbf{Wasserstein distance}: 
Again, we assume a given country and wave $a \leq t \leq b$ is under consideration. For each constituent state $i$, let $\lambda_i$ be the offset that minimises the Wasserstein distance, 
    \begin{align}
    \label{eq:wasserstein}
    W_1(f_i(a:b-\lambda_i),g_i(a+\lambda_i:b)),
\end{align}
where $f_i(a:b-\lambda_i)$ is the distribution associated over the interval $a \leq t \leq b - \lambda_i$, as in (\ref{eq:associateddistribution}), and similarly for $g_i(a+\lambda_i:b)$. Then, let $\lambda_k=[\frac{1}{N}\sum_{i=1}^N\lambda_i]$ be the nearest integer to the mean of the estimated offsets for each state $i$.

\item \textbf{Energy distance}: 
Using similar notation as the above method, for each constituent state $i$, let $\mu_i$ be the offset that minimises the \emph{energy distance} \cite{Szkely2013}, 
\begin{align}
    \label{eq:energy}
    D^2(f_i(a:b-\mu_i),g_i(a+\mu_i:b)),
\end{align}
where $f_i(a:b-\mu_i)$ and $g_i(a+\mu_i:b)$ are distributions defined above and $D^2$ is the $L^2$ integral norm between the associated cumulative distribution functions \cite{Szkely2013}. Then, let $\mu_k=[\frac{1}{N}\sum_{i=1}^N\mu_i]$ analogously as before.

\item \textbf{Normalised inner product}: 
Using similar notation as the above method, for each constituent state $i$, let $\nu_i$ be the offset that minimises the normalised inner product $<.,.>_n$, defined as 
\begin{align}
    \label{eq:innerproduct}
    <x_i(a:b-\nu_i),y_i(a+\nu_i:b)>_n \\ =\frac{x_i(a)y_i(a+\nu_i)+...+x_i(b-\nu_i)y_i(b)}{(x_i(a)^2+...+x_i(b-\nu_i)^2)^\frac{1}{2}(y_i(a+\nu_i)^2+...+y_i(b)^2)^\frac{1}{2}.}
\end{align}
Then, let $\nu_k=[\frac{1}{N}\sum_{i=1}^N\nu_i]$ analogously as before.

\end{enumerate}

Thus we have offsets $\tau_k \in \{\alpha_k, \beta_k, \lambda_k, \mu_k, \nu_k\}$, for each country and wave $k$. Each of these methods considers case and death data on a state-by-state basis, taking into account the federal structure of each country. We remark that the affinity matrix and PDF methods share common features of analysing relationships between different states' proportional sizes of case and death counts. Also, the Wasserstein and energy methods share common features of truncating time series and computing distances between distributions.

Before we present the results of this methodology, we present a proposition that demonstrates our methods work well in the case of simulated data.

\begin{prop}
Let the multivariate time series of cases and deaths for a federation be $x_i(t)$ and $y_i(t)$. Suppose they have the property that there exists a consistent and proportionate progression from cases to deaths after a time lag of $\tau_0$. That is, 
\begin{align}
\label{eq:assumption}
y_i(t)= \begin{cases}
\gamma x_i(t-\tau_0), t=\tau_0+1,...,T\\
0, t \leq \tau_0,
\end{cases}
\end{align}
where $\gamma<1$ and $\tau_0 \in \mathbb{Z}_{>0}$ are constants. Then, for any wave $[T_{k-1},T_k]$ of length at least $\tau_0$, all five methods above return $\tau_k=\tau_0$. That is, all five methods identify the correct offset for the following simulated example.
\end{prop}
\begin{proof}
Let $[T_{k-1},T_k]$ be a fixed interval of length $T^*$. Then the normalised total affinity difference (\ref{eq:affinitydiff}), evaluated for $\alpha=\tau_0$, produces the value
\begin{align}
\frac{1}{T^* - \tau_0} \sum_{t=T_{k-1}}^{T_k - \tau_0} \|\text{Aff}_X(t) - \text{Aff}_Y(t + \tau_0)\|.
\label{eq:affinitydiff2}
\end{align}
By (\ref{eq:assumption}), $y_i(t+\tau_0)=\gamma x_i(t)$ for all $t$ in the interval $[T_{k-1},T_k - \tau_0]$. Thus, $D_Y(t+\tau_0) = \gamma D_X(t)$. Due to the normalisation process of computing the affinity matrix, this implies $\text{Aff}_X(t) = \text{Aff}_Y(t + \tau_0)$ for all $t$. Thus, the normalised total affinity difference for the value $\alpha=\tau_0$ produces the minimal possible value of zero, so the method selects $\alpha_k=\tau_0$.

Next, for the PDF method, the normalised total pdf difference evaluated for $\beta=\tau_0$ produces
\begin{align}
\frac{1}{T^*- \tau_0} \sum_{t=T_{k-1}}^{T_k - \tau_0} \|p_X(t) - p_Y(t + \tau_0)\|_1.
\label{eq:pdfdiff2}
\end{align}
Again by  (\ref{eq:assumption}), we have $y_i(t+\tau_0)=\gamma x_i(t)$ for all $t$ in the interval $[T_{k-1},T_k - \tau_0]$, so $p_X(t)=p_Y(t+\tau_0)$ for all $t \in [T_{k-1},T_k - \tau_0]$. Thus, the normalised total pdf difference for the value $\beta=\tau_0$ produces the minimal possible value of zero, so the method selects $\beta_k=\tau_0$.

Next, we turn to the Wasserstein and Energy distance methods. Here, we can again show that for the selected offset $\lambda_i=\tau_0$, the corresponding Wasserstein distance
    \begin{align}
    \label{eq:wasserstein2}
    W_1(f_i(T_{k-1}:T_k-\tau_0),g_i(T_{k-1}+\tau_0:T_k)),
\end{align}
is equal to zero. Indeed, $x_i(t)_{T_{k-1} \leq t \leq T_k-\tau_0}$ is a scalar multiple of $y_i(t)_{T_{k-1}+\tau_0 \leq t \leq T_k}$, so when both are normalised to distributions $f_i(T_{k-1}:T_k-\tau_0)$ and $g_i(T_{k-1}+\tau_0:T_k)$ respectively, they coincide. Thus, $\lambda_i=\tau_0$ produces the minimal possible value of zero for the Wasserstein distance and so the method selects $\lambda_i=\tau_0$ for each state $i$, hence $\lambda_k=\tau_0$. The same argument holds \textit{mutatis mutandis} for the Energy distance.

Finally, for the normalised inner product method, the same reasoning shows that the normalised inner product achieves its maximal value of 1 when $\nu_i=\tau_0$, so the method selects $\nu_i=\tau_0$ for each state $i$. Hence, $\nu_k$ is analogously chosen to be equal to $\tau_0$.

We remark that the procedure of truncating the interval $[T_{k-1}, T_k]$ to $[T_{k-1}, T_k-\tau_k]$ for the case time series and $[T_{k-1+\tau_k}, T_k]$ for the death time series is essential for the proof to work as above. Indeed, in this simulated example, the death time series $y_i(t)$ has exactly $\tau$ days of leading zeroes before it coincides with a shifted constant times $x_i(t)$, and the truncation is necessary for the methods to select the correct offset.

\end{proof}

\begin{table}
\begin{center}
\begin{tabular}{ |p{3.5cm}|p{1.3cm}|p{1.3cm}|p{1.3cm}|}
 \hline
 Methodology & Wave 1 & Wave 2 & Wave 3 \\
 \hline
 Affinity (US)  & 6 & 37 & 16 \\
 PDF (US) & 5 & 23 & 16 \\
 Wasserstein (US)  & 11 & 19 & 41 \\
  Energy (US)  & 9 & 17 & 38 \\
 Inner product (US)  & 10 & 20 & 29 \\
  Affinity (India)   & 11 & 8 & n/a \\
 PDF (India)   & 8 & 7 & n/a \\
 Wasserstein (India)   & 32 & 5 & n/a \\
  Energy (India)   & 32 & 5 & n/a \\
 Inner product (India)   & 13 & 8 & n/a \\
 Affinity (Brazil)   & 9 & 9 & n/a \\
 PDF (Brazil)   & 9 & 9 & n/a \\
 Wasserstein (Brazil)   & 18 & 13 & n/a \\
 Energy (Brazil)   & 15 & 11 & n/a \\
 Inner product (Brazil)  & 12 & 21 & n/a \\
\hline
\end{tabular}
\caption{Offsets between cases and deaths by country and wave of the pandemic, computed with five different methods, as described in Section \ref{sec:Offsets}. Only the US is determined to have a third wave within our period of analysis.}
\label{tab:Wave_offsets}
\end{center}
\end{table}

Table \ref{tab:Wave_offsets} documents the wave-specific offsets for all three countries among our five methods. We observe broad similarity across all countries and waves between the results obtained by pairs of related methods (affinity and PDF, Wasserstein and energy). Each country presents a unique pattern in the length of their progression from cases to deaths for each wave of the pandemic. First, the US is the only country determined to experience three waves of COVID-19 cases within our analysis window. For all five methods, the first wave produces a significantly lower offset than the second and third waves of COVID-19. The timing of the first wave corresponds to the first half of 2020, when many US states (especially those located in the Northeast) were overwhelmed by early case numbers. As a result, many cases went undetected, and hospitals were unable to administer optimal care to patients. Furthermore, early in the pandemic, there was greater uncertainty within the medical community on suitable treatments for COVID-19 patients.

India, which exhibits two waves of COVID-19 in our analysis window, features almost the opposite observation. As shown in Table \ref{tab:TP1_table}, the length of the first wave in India was $\sim$2.5 times that of the US, and it exhibited a more gradual progression (and subsequent decline) in daily cases until states reached their first peak and trough, respectively. Although much shorter, the second wave was more severe among Indian states - with universally rapid growth in cases and deaths. All five optimisation methods determined the offset of the second wave to be shorter than that of the first wave. This mirrors our finding in the case of the US: when states are overwhelmed with COVID-19, hospitals become overwhelmed with cases, and many patients go undetected - this leads to a decrease in the length of the offset between cases and deaths. This can most likely be explained by latent COVID-19, the inability to access critical equipment (such as ventilators), and inferior treatment within hospitals.

Brazil has quite a different finding again, with little consistency in the offset trend between its first and second waves. Several reasons may explain the variability in our estimates. First, the Brazilian data is quite noisy, with more missing data and reporting issues than the US and India. Second, the variability in the distribution of states' $T_1$ values may suggest limited collective consistency in offset trends among the Brazilian states. Accordingly, we see no clear trend in offset behaviours as we progress from the first to the second wave of the outbreak.

\section{Discussion}
\label{sec:conclusion}

In this paper, we perform a detailed analysis of the three countries most impacted by COVID-19, the US, India and Brazil. Given COVID-19's severe yet varied impact on countries worldwide, our motivation is to understand the differences in the dynamics of the virus' propagation among the world's three worst affected countries. We seek to study both internal structural similarity between states within each country and differences between the countries with respect to several attributes around COVID-19. Comparing the structural dynamics of separate countries' COVID-19 outbreaks may provide insights into the influence different governments, cultures and healthcare systems have had in the evolution of the pandemic. In addition to this explicit contrast, we wanted to explore variability within each country, namely similarity between countries' constituent states. 

First, we study the similarity between case, death and mortality rate trajectories produced by each of our three countries' constituent states. In Section \ref{sec:Trajectories}, we offer methodological contributions as well as non-trivial findings regarding heterogeneity between states in each federation. Our procedure in Algorithm \ref{algorithm} not only identifies a sequence of the most anomalous elements (in this case states) of a collection, it also produces an easily interpretable decreasing curve quantifying the collective heterogeneity. This procedure is robust to the existence of one or even several outlier elements. By the scale of the curves displayed in Figure \ref{fig:Sequential_anomaly_removal_cases}, one can immediately see that India exhibits the greatest heterogeneity between states with respect to the three trajectories analysed, particularly rolling mortality rates. This is a robust finding that consistently holds even when we remove anomalous states, and highly non-trivial given the findings of Section \ref{sec:Turningpoints} discussed below. The specific identification of the most anomalous states is also non-obvious, revealing different patterns in each federation. In the US, we find that the most anomalous behaviour is consistently located in the Northeast. In India, the state Lakshadweep is consistently identified as most anomalous in cases, deaths and mortality. In Brazil, there is less consistency in the type of anomalies identified among our three attributes.

The insights generated above concern broad structure in the data on a state-by-state basis. We have combined existing statistical learning methodologies (such as clustering), a new distance between trajectories as well as a new algorithmic approach to identify specific states and quantify overall heterogeneity, with robustness to outliers. The insights presented in this manuscript would not be possible without a combination of existing (rather sophisticated) and new (rather bespoke) procedures, all carefully considered for the application. More broadly, most COVID-19 data consumed by the general public is reported at the national level; most variation within states is ignored, especially a detailed quantification of heterogeneity. Our methods combine non-trivial mathematical investigation with data sets that are typically not examined in detail at the state level.

In Section \ref{sec:Turningpoints}, we apply our turning point algorithm to study wave behaviours among the three countries. In the US, where three waves of COVID-19 cases are observed, a median first wave length of 92 days is found among the distribution of US states. By contrast, Indian states produced a median first wave length of 231 days, with a lower variance than the US, and just two waves of COVID-19 cases overall. In Brazil, where two waves of the cases were also identified, the median length of states' first wave was 143, with high variance. Our analysis suggests that US and Indian states exhibit stronger characteristic behaviours than those exhibited by Brazil. Indeed, clustering reveals that the US and India are quite dissimilar in wave behaviour, almost entirely clustering among themselves, while Brazil is quite heterogeneous, with some states similar to US states, some similar to Indian states, and some outlier states.

These findings are highly non-trivial without undertaking judicious mathematical analysis as we have done. Numerous papers on COVID-19 simply estimate the duration of waves by inspection or other unreliable methods, while we use a careful algorithm to do so. Unlike most work, we do so on a state-by-state basis, and thus must deal with data issues such as anomalous counts and missing values. Our findings contrast notably with Section \ref{sec:Trajectories} and are highly non-trivial to guess. While it is predictable that US and Indian states exhibit relatively strong characteristic wave behaviours among themselves, it is certainly non-trivial that Brazilian states interleave between US and Indian states with respect to wave behaviour, and that the distribution of first wave length among Brazilian states (Figure \ref{fig:Distribution_TP_1}) is so broad. Further, it is striking that Section \ref{sec:Trajectories} reveals the greatest heterogeneity between Indian states in terms of trajectories, but Section \ref{sec:Turningpoints} demonstrates the least variance in first wave length (Table \ref{tab:TP1_table}). This is not necessarily contradictory but is highly non-obvious: case and death curves exhibit substantial differences but the overall wave pattern is more uniform across India.


Finally, Section \ref{sec:Offsets} introduces new optimisation methodologies to study the progression of COVID-19 cases to deaths in each of our three countries' waves of the pandemic. We believe this is the first work to explicitly acknowledge that the progression from cases to deaths may vary between different waves of the pandemic and aim to study this. In the US, we highlight a significantly longer period between diagnosis and death in the second and third waves of COVID-19 cases. This finding is consistent among all five optimisation methods. In India, all five methods demonstrate a sharp reduction in the length of this offset as we progress from the first to the second wave. In Brazil, we find limited consistency among our methods, with no clear takeaway regarding the change in the length of the COVID-19 case life cycle, in the first and second waves. In aggregate, our analysis suggests that when countries become overwhelmed with COVID-19 cases, the length of the case-to-death progression decreases. This may be due to overwhelmed hospital systems, sub-optimal medical treatment, limited access to medical resources such as ventilators and an increase in undetected cases. We also include theoretical validation of our methodology, which is non-trivial due to the truncation of time series inherent in the case and death data (that is, death data lag behind cases and non-zero counts begin later).

There are several reasons why these determinations of offsets between cases and deaths are not particularly obvious. First, they are computed in a high dimensional manner with several methods that use the federal structure of the three countries. Second, the changes between waves of these offsets are different for all three federations, which we believe shows the impossibility of a straightforward prediction of their behaviour. Algorithmic techniques must be used to identify time series turning points (corresponding to waves of the pandemic), and the relationship between cases and deaths is fluid - varying over time, across countries and between countries' constituent states and territories. Although the offset in the progression from COVID-19 cases to deaths is only one facet of a hugely complex global pandemic, it is of great importance to understand for the future treatment and management of COVID-19 cases. COVID-19 data follows a causal structure: any COVID-19 case will ultimately progress into either the recovered or death category. This causal structure is typically modelled via SIRD models and their variants described in Section \ref{sec:Introduction}. These have their utility, but are not ideal to study the multi-wave dynamics of COVID-19 brought about by regularly shifting government restrictions and community behaviour. We choose to exclusively address the transition from cases to deaths without the strong parametric assumptions in SIRD models; we believe this progression to be of direct importance in treating COVID-19 patients currently burdening many countries' healthcare systems.

\subsection{Future work}

There are many avenues for potential future work, in both methodological and applied contexts. First, one could investigate the reasons for more or less heterogeneity among constituent states for various countries. For example, one could explore why Brazil's states experienced rather different outcomes relative to wave behaviours and progression from cases to deaths. In this paper, we highlight that these differences are far more significant than the USA and India. Indeed, Brazil's human development index (HDI) of 0.765 is between that of the US (0.926) and India (0.645), and it is conceivable that development among Brazilian states differs more than that among the US or India. This, along with other predictors, may help construct supervised and unsupervised learning algorithms where relationships can be learned and associations can be formed, respectively.

Next, the methods that are introduced in this paper could be extended. Although the offsets in this paper have been implemented in discrete time partitions, these methods could conceivably be implemented in a rolling manner, where a continuous (time-varying) offset may be estimated. Furthermore, the theoretical aspects of these estimators could be further investigated, and tested on data generated from a variety of data generating processes. This may include noise generated from a wide variety of distributions, adversarial data such as extreme points and outliers, and so on. In addition, future work could further explore the aforementioned causal structure in the data, including offsets between time series of COVID-19 cases, counts of recovered patients (including those who experience ``long Covid'' \cite{Mahase2020}) and COVID-19 deaths. One could compare the offsets between COVID-19 cases and deaths, and COVID-19 cases and recovered patients separately - and then study whether there is a latent relationship between these two offsets, and more specifically, study how they evolve with time. Our descriptive and nonparametric analysis could conceivably be incorporated with judiciously chosen SIRD models on a wave by wave basis.

At the time of writing this paper, many parts of the world are currently experiencing a fourth wave of COVID-19 cases. Many European countries such as Austria and Germany are attracting a substantial amount of publicity, regarding their growth in new daily COVID-19 cases. It would be of great interest to compare the heterogeneity of COVID-19 epidemiology within differing states or regions of these countries, and estimate the offset in the progression from cases to deaths during the fourth wave of the pandemic. In particular, with the appropriate data, one could distinguish between the vaccinated and unvaccinated populations. 

\section{Conclusion}
\label{sec:finalconclusion}

Overall, we have identified numerous features that characterise the nature of the pandemic within the US, India and Brazil. India exhibits the greatest heterogeneity in its trajectories, and yet simultaneously the most homogeneity in its wave behaviours due to a very long first wave and a rapid second wave in almost every state. The US and India cluster quite separately in trajectory and wave behaviours, while Brazilian states are interleaved between them, characterised by the greatest variance in wave lengths. A similar distinction is observed in offsets, where the US case-to-death progressions drastically lengthen between first and subsequent waves, the reverse holds for India, while Brazil is again a mixture of the two.

Throughout this work, we have identified specific states within the three federations as the most anomalous and determined various non-trivial features in the federations' COVID-19 behaviour, including heterogeneity of trajectories, wave behaviour, and the progression from cases to deaths. New methodologies have been presented for this purpose, including the ability to more robustly determine distances between trajectories and determine patterns in overall heterogeneity without too much vulnerability to outliers. We have identified numerous avenues for future work to apply these methods in new contexts, such as Europe's fourth wave, or to undertake closer analysis with researchers from other disciplines to investigate some of the policy measures or regional features that could be contributing to these patterns.


\section*{Data availability}
Daily COVID-19 case and death counts for the US, India and Brazil can be found at the New York Times \cite{datasourcenyt}, PRS Legislative Research \cite{indiadata} and the Brazilian Ministery of Health \cite{brazildata}, respectively.

\section*{Funding sources}
This research did not receive any specific grant from funding agencies in the public, commercial, or not-for-profit sectors.

\appendix

\section{Turning point methodology}
\label{appendix:turningpoint}

In this section, we provide more details for identifying turning points of a new case time series $x(t)$. First, some smoothing of the counts is necessary due to data irregularities and discrepancies between different data sources. There are consistently lower counts on the weekends and some negative counts due to retroactive adjustments. A Savitzky-Golay filter ameliorates these issues by combining polynomial smoothing with a moving average computation - this moving average eliminates all but a few small negative counts; we then replace these negative counts with zero. This yields a smoothed time series $\hat{x}(t) \in \mathbb{R}_{\geq 0}.$ Subsequently, we perform a two-step process to select and then refine a non-empty set $P$ of local maxima (peaks) and $T$ of local minima (troughs).

Following \cite{james2020covidusa}, we apply a two-step algorithm to the smoothed time series $\hat{x}(t)$. The first step produces an alternating sequence of troughs and peaks, beginning with a trough at $t=1$, when there are zero cases. The second step refines this sequence according to chosen conditions and parameters. The primary conditions to identify a peak or trough, respectively, in the first step, are the following:
\begin{align}
\label{baddefnpeak}
\hat{x}(t_0)&=\max\{\hat{x}(t): \max(1,t_0 - l) \leq t \leq \min(t_0 + l,T)\},\\
\label{baddefntrough}\hat{x}(t_0)&=\min\{\hat{x}(t): \max(1,t_0 - l) \leq t \leq \min(t_0 + l,T)\},
\end{align}
where $l$ is a parameter to be chosen. Following \cite{james2020covidusa}, we select $l=17$, which accounts for the 14-day incubation period of the virus \cite{incubation2020} and less testing on weekends. Defining peaks and troughs according to this definition alone has several flaws, including the potential for two consecutive peaks.

Instead, we implement an inductive procedure to select an alternating sequence of peaks and troughs. Suppose $t_0$ is the last determined peak. We search in the period $t>t_0$ for the first of two cases: if we find a time $t_1>t_0$ that satisfies (\ref{baddefntrough}) and a non-triviality condition $\hat{x}(t_1)<\hat{x}(t_0)$, we add $t_1$ to the set of troughs and proceed from there. If we find a time $t_1>t_0$ that satisfies (\ref{baddefnpeak}) and  $\hat{x}(t_0)\geq \hat{x}(t_1)$, we ignore this lower peak as redundant; if we find a time $t_1>t_0$ that satisfies (\ref{baddefnpeak}) and  $\hat{x}(t_1) > \hat{x}(t_0)$, we remove the peak $t_0$,  replace it with $t_1$ and proceed from $t_1$. A similar process applies from a trough at $t_0$. 

At this point, a time series is assigned an alternating sequence of troughs and peaks. However, some turning points are immaterial and should be excluded. The second step is a flexible approach introduced in \cite{james2020covidusa} for this purpose. In this paper, we introduce new conditions within this framework. First, let $t_m$ be the global maximum of $\hat{x}(t)$. If this is not unique, we declare $t_m$ to be the first global maximum. This point $t_m$ is always declared a peak during the first step detailed above. Given any other peak $t_1$, we compute the peak ratio $\frac{\hat{x}(t_1)}{\hat{x}(t_m)}$. We select a parameter $\delta$, and if $\frac{\hat{x}(t_1)}{\hat{x}(t_m)}<\delta$, we remove the peak $t_1$. If two consecutive troughs $t_0,t_2$ remain, we remove $t_0$ if $\hat{x}(t_0)>\hat{x}(t_2)$, and remove $t_2$ if $\hat{x}(t_0)\leq\hat{x}(t_2)$. That is, we ensure the sequence of peaks and troughs remains alternating. In our implementation, we choose $\delta=0.01.$ Unlike \cite{james2020covidusa}, we remove earlier peaks, not just subsequent peaks, according to this condition.

Finally, we use the same \emph{log-gradient} function between times $t_1<t_2$, defined as
\begin{align}
\label{loggrad}
   \loggrad(t_1,t_2)=\frac{\log \hat{x}(t_2) - \log \hat{x}(t_1)}{t_2-t_1}.
\end{align}
The numerator equals  $\log(\frac{\hat{x}(t_2)}{\hat{x}(t_1)})$, a "logarithmic rate of change." Unlike a standard rate of change given by $\frac{\hat{x}(t_2)}{\hat{x}(t_1)} -1$, the logarithmic change is symmetrically between $(-\infty,\infty)$. Let $t_1,t_2$ be adjacent turning points (one a trough, one a peak). We choose a parameter $\epsilon=0.01$;  if
\begin{align}
    |\loggrad(t_1,t_2)|<\epsilon,
\end{align}
that is, the average logarithmic change is less than 1\%, we remove $t_2$ from our sets of peaks and troughs. If $t_2$ is not the final turning point, we also remove $t_1$.

\bibliographystyle{_elsarticle-num-names}
\bibliography{__References}
\end{document}